\documentclass[runningheads]{llncs}

\usepackage{amsmath,amssymb,amsfonts}
\usepackage{graphicx}
\usepackage{textcomp}
\usepackage{xcolor}

\usepackage[ruled, vlined, nofillcomment, linesnumbered]{algorithm2e}
\usepackage{url}
\usepackage{graphicx}
\usepackage{subcaption}
\usepackage[english]{babel}
\usepackage{booktabs}
\usepackage{verbatim}
\usepackage{float}
\usepackage{footmisc}
\usepackage{xspace}
\usepackage{tikz}
\usepackage{pgfplots}
\usepackage{pgfplotstable}

\usetikzlibrary{fit,external,positioning}
\usetikzlibrary{shapes.geometric}

\usepackage[numbers,sort&compress]{natbib}

\usepackage{etoolbox}

\newcommand{\set}[1]{\left\{#1\right\}}
\newcommand{\pr}[1]{\left(#1\right)}
\newcommand{\fpr}[1]{\mathopen{}\left(#1\right)}

\newcommand{\fspr}[1]{\mathopen{}\left[#1\right]}

\newcommand{\abs}[1]{{\left|#1\right|}}

\newcommand{\np}{\textbf{NP}}

\newcommand{\define}{\leftarrow}

\DeclareRobustCommand{\dispfunc}[2]{%
  \ensuremath{%
  \ifthenelse{\equal{#2}{}}%
    {\mathit{#1}}%
    {\mathit{#1}\fpr{#2}}}}

\newcommand{\dens}[1]{\dispfunc{d}{#1}}
\newcommand{\jaccard}[1]{\dispfunc{J}{#1}}

\newcommand{\score}[1]{\dispfunc{q}{#1}}

\newcommand{\bigO}[1]{\dispfunc{\mathcal{O}}{#1}}

\newcommand{\matchprb}{\textsc{3DM-2}\xspace}

\newcommand{\algsoftone}{\textsc{Itr}\xspace}
\newcommand{\algsofttwo}{\textsc{Grd}\xspace}

\newcommand{\fm}[1]{\mathcal{#1}}

\newcommand{\mean}[2]{\operatorname{E}_{#1}\fspr{#2}}

\newcommand{\dtname}[1]{\textsl{#1}}

\newcommand{\calG}{\ensuremath{{\mathcal G}}}
\newcommand{\calS}{\ensuremath{{\mathcal S}}}

\newcommand{\problemsoft}{\textsc{JWDS}\xspace}

\SetKwComment{tcpas}{\{}{\}}
\SetCommentSty{textnormal}
\SetArgSty{textnormal}
\SetKw{False}{false}
\SetKw{True}{true}
\SetKw{Null}{null}
\SetKwInOut{Output}{output}
\SetKwInOut{Input}{input}
\SetKw{AND}{and}
\SetKw{OR}{or}
\SetKw{Break}{break}

\SetKwFor{Until}{until}{do}{}

\definecolor{yafcolor5}{rgb}{0.141, 0.345, 0.643}
\SetKw{Output}{output}

\DeclareRobustCommand{\dispfunc}[2]{%
	\ensuremath{%
		\ifthenelse{\equal{#2}{}}%
			{\mathit{#1}}%
			{\mathit{#1}\fpr{#2}}}}

\SetKw{Break}{break}

\pgfdeclarelayer{background}
\pgfdeclarelayer{foreground}
\pgfsetlayers{background,main,foreground}

\definecolor{yafaxiscolor}{rgb}{0.3, 0.3, 0.3}

\definecolor{yafcolor1}{rgb}{0.4, 0.165, 0.553}
\definecolor{yafcolor2}{rgb}{0.949, 0.482, 0.216}
\definecolor{yafcolor3}{rgb}{0.47, 0.549, 0.306}
\definecolor{yafcolor4}{rgb}{0.925, 0.165, 0.224}
\definecolor{yafcolor5}{rgb}{0.141, 0.345, 0.643}
\definecolor{yafcolor6}{rgb}{0.965, 0.933, 0.267}
\definecolor{yafcolor7}{rgb}{0.627, 0.118, 0.165}
\definecolor{yafcolor8}{rgb}{0.878, 0.475, 0.686}

\newlength{\yafaxispad}
\newlength{\yaftlpad}
\newlength{\yaflabelpad}
\newlength{\yafaxiswidth}
\newlength{\yafticklen}
\makeatletter
\def\pgfplots@drawtickgridlines@INSTALLCLIP@onorientedsurf#1{}
\makeatother

\newcommand{\yafdrawaxis}[4]{
	\pgfplotstransformcoordinatex{#1}\let\xmincoord=\pgfmathresult 
	\pgfplotstransformcoordinatex{#2}\let\xmaxcoord=\pgfmathresult 
	\pgfplotstransformcoordinatey{#3}\let\ymincoord=\pgfmathresult 
	\pgfplotstransformcoordinatey{#4}\let\ymaxcoord=\pgfmathresult 
	\pgfsetlinewidth{\yafaxiswidth} 
	\pgfsetcolor{yafaxiscolor}
	\pgfpathmoveto{\pgfpointadd{\pgfpointadd{\pgfplotspointrelaxisxy{0}{0}}{\pgfqpointxy{\xmincoord}{0}}}{\pgfqpoint{-0.5\yafaxiswidth}{\yafaxispad}}}
	\pgfpathlineto{\pgfpointadd{\pgfpointadd{\pgfplotspointrelaxisxy{0}{0}}{\pgfqpointxy{\xmaxcoord}{0}}}{\pgfqpoint{0.5\yafaxiswidth}{\yafaxispad}}}
	\pgfpathmoveto{\pgfpointadd{\pgfpointadd{\pgfplotspointrelaxisxy{0}{0}}{\pgfqpointxy{0}{\ymincoord}}}{\pgfqpoint{\yafaxispad}{-0.5\yafaxiswidth}}}
	\pgfpathlineto{\pgfpointadd{\pgfpointadd{\pgfplotspointrelaxisxy{0}{0}}{\pgfqpointxy{0}{\ymaxcoord}}}{\pgfqpoint{\yafaxispad}{0.5\yafaxiswidth}}}
	\pgfusepath{stroke}
}

\pgfplotscreateplotcyclelist{yaf}{%
{yafcolor5,mark options={scale=0.75},mark=o}, 
{yafcolor2,mark options={scale=0.75},mark=square},
{yafcolor3,mark options={scale=0.75},mark=triangle},
{yafcolor4,mark options={scale=0.75},mark=o},
{yafcolor1,mark options={scale=0.75},mark=o},
{yafcolor6,mark options={scale=0.75},mark=o},
{yafcolor7,mark options={scale=0.75},mark=o},
{yafcolor8,mark options={scale=0.75},mark=o}} 

\pgfkeys{/pgf/number format/.cd,1000 sep={\,}}

\pgfplotsset{axis y line=left, axis x line=bottom,
	tick align=outside,
	tickwidth=\yafticklen,
	clip = false,
    x axis line style= {-, line width = 0pt, color=black!0},
    y axis line style= {-, line width = 0pt, color=black!0},
    x tick style= {line width = \yafaxiswidth, color=yafaxiscolor, yshift = \yafaxispad},
    y tick style= {line width = \yafaxiswidth, color=yafaxiscolor, xshift = \yafaxispad},
    x tick label style = {font=\small, yshift = \yaftlpad, inner xsep = 0pt},
    y tick label style = {font=\small, xshift = \yaftlpad},
    every axis y label/.style = {at = {(ticklabel cs:0.5)}, rotate=90, anchor=center, font=\small, yshift = -\yaflabelpad, inner sep = 0pt},
    every axis x label/.style = {at = {(ticklabel cs:0.5)}, anchor=center, font=\small, yshift = \yaflabelpad},
    x tick label style = {font=\small, yshift = 1pt},
    grid = major,
    major grid style  = {dash pattern = on 1pt off 3 pt},
	every axis plot post/.append style= {line width=\yafaxiswidth} ,
	legend cell align = left,
	legend style = {inner sep = 1pt, cells = {font=\scriptsize}},
	legend image code/.code={%
		\draw[mark repeat=2,mark phase=2,#1] 
		plot coordinates { (0cm,0cm) (0.15cm,0cm) (0.3cm,0cm) };%
	} 
}

\newcommand{\pgfprintduration}[1]{%
	\ifthenelse{\equal{#1}{}}{---}{%
	\pgfmathsetmacro{\minutes}{floor(#1 / 60)}%
	\pgfmathsetmacro{\seconds}{#1 - 60*\minutes}%
	\pgfmathifthenelse{\minutes > 0}{"\pgfmathprintnumber{\minutes}m \pgfmathprintnumber[fixed,precision=0]{\seconds}s"}{"\pgfmathprintnumber{\seconds}s"}\pgfmathresult}}

\begin{document}

\title{Jaccard-constrained dense subgraph discovery}

\author{Chamalee Wickrama Arachchi \and Nikolaj Tatti} 

\institute{HIIT, University of Helsinki, firstname.lastname@helsinki.fi}



\maketitle

\begin{abstract}
Finding dense subgraphs is a core problem in graph mining with many
applications in diverse domains. At the same time 
many real-world networks vary over time, that is, the dataset can be represented as a sequence of graph snapshots.
Hence, it is natural to consider the question of finding dense subgraphs in a temporal network that are allowed to vary over time to a certain degree.
In this paper, we search for dense subgraphs that have large pairwise Jaccard similarity coefficients.
More formally,
given a set of graph snapshots and a weight $\lambda$, we find a collection of dense subgraphs such that
the sum of densities of the induced subgraphs plus the sum of Jaccard indices, weighted by $\lambda$, is maximized.
We prove that this problem is  \np-hard. To discover dense subgraphs with good objective value, we present an iterative algorithm which runs in  $\bigO{n^2k^2 + m \log n + k^3
n}$ time per single iteration, and a greedy algorithm which runs in
$\bigO{n^2k^2 + m \log n + k^3 n}$ time, where $k$ is
the length of the graph sequence and $n$ and $m$ denote number of nodes and
total number of edges respectively. 
We show experimentally that our algorithms are efficient, they can find ground truth in synthetic datasets and provide interpretable results
from real-world datasets. Finally, we present a case study that shows the usefulness of our problem.

\end{abstract}


\section{Introduction}

Finding dense subgraphs is a core problem in graph mining with many applications in diverse domains such as social network analysis~\cite{semertzidis2019finding}, temporal pattern mining in financial markets~\cite{xiaoxi2009migration}, and  biological system analysis~\cite{fratkin2006motifcut}. 
Often, many real-world networks
vary over time, in which case a sequence of graph snapshots naturally exists. Consequently, mining dense subgraphs over time  has gained an attention in data mining literature~\cite{jethava2015finding,semertzidis2019finding,rozenshtein2020finding,galimberti2020core}. 

Our goal is to find dense subgraphs in a temporal network.
In order to measure density, we
will use a popular choice of the ratio between number of induced edges and nodes.
This choice is popular since the densest subgraph can be found in polynomial time~\citep{goldberg1984finding} and approximated efficiently~\citep{charikar2000greedy}.

Given a graph sequence, there are natural extremes to find the densest subgraphs: the first approach is to find a common subgraph that maximizes the sum of densities for individual snapshots,
as proposed by~\citet{semertzidis2019finding} among other techniques. The other approach is to find the densest subgraphs for each snapshot individually.

In this paper, we study the problem that bridges the gap between these two extremes,
namely we seek
dense subgraphs in a temporal network
that are allowed to vary over time to a certain degree.
Our approach is to incorporate the Jaccard similarity index directly
into our objective function along with the density. 
Here, we reward similar graphs over the snapshots.
More formally,
given a graph sequence $\calG$ and parameter
$\lambda$, we seek a sequence of subgraphs,  such that the sum of densities
plus the sum of Jaccard indices, weighted by $\lambda$,  
is maximized.

We demonstrate the objective in the following toy example.

\begin{example}
\label{ex:toy-soft}

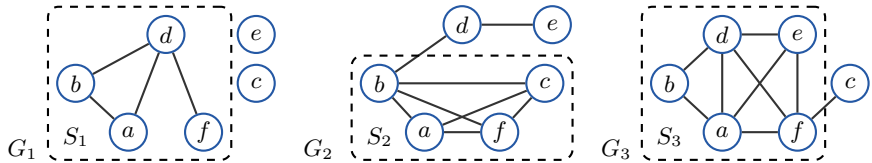
\begin{figure}[t]
\tikzstyle{exnode} = [thick, draw = yafcolor5, circle, inner sep = 1pt, text=black, minimum width=14pt]
\tikzstyle{exedge} = [black!80, thick]
\tikzstyle{exbox} = [dashed, black, thick, draw, rounded corners]
\hspace*{\fill}
\begin{tikzpicture}[baseline = 0pt, yscale=0.65]
\node[anchor = east, inner sep = 0pt] at (-0.7, -1.35) {$G_1$};
\node[exnode] (n1) at (0.5, -1) {$a$};
\node[exnode] (n2) at (-0.2, 0) {$b$};
\node[exnode] (n3) at (2.2, 0) {$c$};
\node[exnode] (n4) at (1, 1) {$d$};
\node[exnode] (n8) at (2.2, 1) {$e$};

\node[exnode] (n5) at (1.5, -1) {$f$};

\node[inner sep = 0pt] (l1) at (-0.2, -1.1) {$S_1$};

\draw[exedge, bend left = 0] (n1) edge (n2);
\draw[exedge, bend left = 0] (n1) edge (n4);
\draw[exedge, bend left = 0] (n2) edge (n4);

\draw[exedge, bend left = 0] (n4) edge (n5);


\node[exbox, fit = (n1) (n2) (n4) (l1) (n5)] (b1) {};
\end{tikzpicture}
\hfill
\begin{tikzpicture}[baseline = 0pt, yscale=0.65]
\node[anchor = east, inner sep = 0pt] at (-0.7, -1.35) {$G_2$};
\node[exnode] (n1) at (0.5, -1) {$a$};
\node[exnode] (n2) at (-0.1, 0) {$b$};
\node[exnode] (n3) at (2.1, 0) {$c$};
\node[exnode] (n4) at (1, 1.2) {$d$};
\node[exnode] (n8) at (2.2, 1.2) {$e$};

\node[exnode] (n5) at (1.5, -1) {$f$};

\node[inner sep = 0pt] (l1) at (-.1, -1.1) {$S_2$};

\draw[exedge, bend left = 0] (n1) edge (n2);
\draw[exedge, bend left = 0] (n1) edge (n3);
\draw[exedge, bend left = 0] (n1) edge (n5);
\draw[exedge, bend left = 0] (n2) edge (n3);
\draw[exedge, bend left = 0] (n2) edge (n4);
\draw[exedge, bend left = 0] (n2) edge (n5);
\draw[exedge, bend left = 0] (n4) edge (n8);

\draw[exedge, bend left = 0] (n5) edge (n3);

\node[exbox, fit = (n1) (n2) (n3)  (l1) (n5)] (b1) {};
\end{tikzpicture}
\hfill
\begin{tikzpicture}[baseline = 0pt, yscale=0.65]
\node[anchor = east, inner sep = 0pt] at (-0.7, -1.35) {$G_3$};
\node[exnode] (n1) at (0.5, -1) {$a$};
\node[exnode] (n2) at (-0.2, 0) {$b$};
\node[exnode] (n3) at (2.2, 0) {$c$};
\node[exnode] (n4) at (0.5, 1) {$d$};
\node[exnode] (n8) at (1.5, 1) {$e$};

\node[exnode] (n5) at (1.5, -1) {$f$};

\node[inner sep = 0pt] (l1) at (-.2, -1.1) {$S_3$};

\draw[exedge, bend left = 0] (n1) edge (n2);
\draw[exedge, bend left = 0] (n1) edge (n4);
\draw[exedge, bend left = 0] (n1) edge (n5);
\draw[exedge, bend left = 0] (n1) edge (n8);
\draw[exedge, bend left = 0] (n2) edge (n4);

\draw[exedge, bend left = 0] (n4) edge (n5);
\draw[exedge, bend left = 0] (n4) edge (n8);

\draw[exedge, bend left = 0] (n5) edge (n3);
\draw[exedge, bend left = 0] (n5) edge (n8);

\node[exbox, fit = (n1) (n2)  (n4) (l1)  (n5) ] (b1) {};
\end{tikzpicture}
\hspace*{\fill}
\caption{Toy graphs used in Example \ref{ex:toy-soft}}
\label{fig:toy}
\end{figure}

Consider a graph sequence $\{G_1,G_2,G_3\}$ shown in Figure~\ref{fig:toy}, 
each graph consisting of $6$ vertices and varying edges. 
We denote  the density induced by the vertex set $S_i$  by $\dens{S_i}$, defined as
the ratio between number of induced edges and nodes, $\dens{S_i} = \frac{\abs{E(S_i)}}{\abs{S_i}}$.

Given a weight parameter $\lambda$ and a sequence of subgraphs $\set{S_1, S_2, S_3}$, we define our objective function as
$ \sum_{i=1}^{3} \dens{S_i} + \lambda \sum_{i < j} \jaccard{S_i, S_j}$. Note that Jaccard index between set $S$ and $T$ is defined as $\jaccard{S, T} = \frac{|S \cap T|}{|S \cup T|}$.
	
Assume that we set $\lambda = 0.3$
and select $S_1 = \{a,b,d,f\}$, $S_2 = \{a,b,c,f\}$, and $S_3 = \{a,b,d,e,f\}$.
The sum of densities is $\frac{4}{4} + \frac{6}{4} + \frac{8}{5} = 4.1$. 
Here, $\jaccard{S_1, S_2} = \frac{3}{5} = 0.6$, $\jaccard{S_1, S_3} = \frac{4}{5} = 0.8$, and $\jaccard{S_2, S_3} = \frac{3}{6} = 0.5$. 
Therefore, our objective is equal to $4.1 + 0.3 \times (0.6 + 0.8 + 0.5) = 4.67$.
\end{example}

We show that our problem is \np-hard and consequently propose two greedy algorithms.
The first approach is an iterative algorithm
where we start either with the common densest subgraph or a set of the densest subgraphs for each individual snapshot and then iteratively  improve each individual snapshot. The improvement step is done with a classic technique used to approximate
the densest subgraph (in a single graph)~\citep{asahiro2000greedily,charikar2000greedy}. We start with the complete snapshot,
and iteratively seek out the vertex so that the remaining graph yields the largest score.
We remove the vertex, and the procedure is repeated until no vertices remain; the best subgraph for
that snapshot is selected. This is repeated for each snapshot until no improvement is possible.
The second algorithm uses a similar approach except we consider all snapshots at the same time,
and select the best subgraph once we are done.

The appeal of this approach is that, when dealing with a single graph, finding the next vertex
can be done efficiently using a priority queue~\citep{charikar2000greedy,asahiro2000greedily}. We cannot use this approach directly due to the updates
in Jaccard indices. Instead we maintain a \emph{set} of priority queues that allow us to find vertices
quickly in practice.

The remainder of the paper is organized as follows.
In Section \ref{sec:prel}, we  provide preliminary notation along with the
formal definitions of our optimization problem. Next, we prove
\np-hardness of our problem in Section~\ref{sec:compl}. All our
algorithms and their running times are  presented in Section~\ref{sec:algorithm}. Related work is
discussed in Section~\ref{sec:related}. Section~\ref{sec:exp} contains an
experimental study both with synthetic and real-world datasets. 
Finally, Section~\ref{sec:conclusions} summarizes the paper and
provides directions for the future work.

\section{Preliminary notation and problem definition}\label{sec:prel}

We begin by providing preliminary notation and formally defining our problem.

Our input is a sequence of graphs $\calG = \set{G_1, \ldots, G_k}$, where each snapshot $G_i = (V, E_i)$
is defined over the same set of nodes. We often denote the number of nodes and edges by $n = \abs{V}$ and 
$m_i = \abs{E_i}$, or $m = \abs{E}$, if $i$ is omitted.

Given a graph $G = (V, E)$, and a set of nodes $S \subseteq V$, we define $E(S) \subseteq E$ to be the subset of edges
having both endpoints in $S$.

As mentioned before, our goal is to find dense subgraphs in a temporal network, and for that
we need to quantify the density of a subgraph. More formally,
assume an unweighted graph $G = (V, E)$, and let $S \subseteq V$.
We define the \emph{density} $\dens{S}$ of a single node set $S$
and  extend this definition for a sequence of subgraphs $\calS = \set{S_1, \ldots, S_k}$
by writing
\[
	\dens{S_i} = \frac{\abs{E(S_i)}}{\abs{S_i}}
	\quad\quad \text{and}\quad\quad
	\dens{\calS} = \sum_{i = 1}^k \dens{S_i}
	\quad.
\]


We will use the Jaccard index in order to measure the similarity between two subgraphs. More formally,
given two sets of nodes $S$ and $T$, we write
\[
	\jaccard{S, T} = \frac{|S \cap T|}{|S \cup T|}\quad. 
\]

Ideally, we would like to have each subgraph to have high density, and share as many nodes as possible with each other. 
This leads  to the following score and optimization problem.
More specifically, given a weight parameter $\lambda$ and a sequence of subgraphs $\calS = \set{S_1, \ldots, S_k}$ we define a score
\[
	\score{\calS; \lambda} = \dens{\calS} + \lambda \sum_{i=1}^{k} \sum_{j = i + 1}^{k} \jaccard{S_i, S_j} \quad .
\]

\begin{problem}[Jaccard Weighted Densest Subgraphs ({\problemsoft})]
Given a graph sequence $\calG = \set{G_1,\dots,G_k}$, with  $G_i = (V, E_i)$,
and a real number $\lambda$, find a collection of subset of vertices
$\calS = \set{S_1,\dots,S_{k}}$,
such that $\score{\calS}$ is maximized.
\end{problem}

We will consider two extreme cases. The first case is when $\lambda$ is very large, say $\lambda = \sum m_i$, which we refer to
as \emph{densest common subgraph} or DCS. This problem can be solved by first flattening the graph sequence into one weighted graph,
where an edge weight is the number of snapshots in which an edge occurs.
The problem is then a standard (weighted) densest subgraph problem that can be solved
using the method given by~\citet{goldberg1984finding} in $\bigO{nm \log n}$ time.
The other extreme case is $\lambda = 0$ which can be found by solving the densest subgraph
problem for each individual snapshot.

The main difference from prior
studies~\cite{jethava2015finding,semertzidis2019finding} is that we allow the
subsets to be varied within a given  margin~(which is defined by Jaccard
coefficient), without enforcing subsets to be fully identical.

\section{Computational complexity}\label{sec:compl}
The problem of finding a common subgraph which maximizes the sum of densities can be solved optimally in
polynomial time. Moreover, if we set $\lambda = 0$, then we can solve the problem by 
finding optimal dense subgraphs for each snapshot individually.
Next we show that \problemsoft is \np-hard.
Note that the hardness relies on the fact that we can choose a specific $\lambda$.

\begin{proposition} 
\label{prop:np2}
\problemsoft is \np-hard.
\end{proposition} 
\begin{proof}
We will prove the hardness by reducing from a 2-bounded 3-set packing problem
\matchprb, a problem where we are given a set of items $U$, a family of sets $\fm{C}$ each of size $3$
such that each item in $U$ is included in exactly two sets, and are asked to find a maximum matching~\citep{chlebik2003approximation}.

Assume that we are given an instance with $r = \abs{\fm{U}}$ items and $\ell = \abs{\fm{C}}$ sets.
For each set $C_i$ we introduce two nodes $v_i$ and $v'_i$, and for each item $u_i$ we introduce a node $w_i$.
We also introduce four additional nodes $z_1$, $z_2$, $z_3$, and $z_4$. In total, we have $n = 2\ell + r + 4$ nodes.

Let $u_i$ be an item and let $C_a$ and $C_b$ be two sets containing $u_i$.
We add two snapshots: the first graph $G_i$ contains $(z_1, z_2)$, $(z_1, z_3)$, $(w_i, z_1)$, $(v_a, z_1)$, and $(v_b, z_1)$ edges
and the second graph $G_i'$ contains $(z_1, z_2)$, $(z_1, z_3)$, $(w_i, z_1)$, $(v_a', z_1)$, and $(v_b', z_1)$ edges.

We also add $q = 2rn^4$ graphs $F_j$, each with three edges $(z_1, z_2)$, $(z_1, z_3)$, and $(z_1, z_4)$. We set $\lambda = 0.55/q$.

Let $\calS$ be the optimal solution.
We claim that
\begin{equation}
\begin{split}
\label{eq:cond}
\score{\calS} \geq T_1 + T_2 + T_3, \text{ where } & T_1 = q \frac{3}{4} + \lambda {q \choose 2},\ T_2 =  2r \times 1.08,\ \text{ and } \\
& T_3 = 3 \lambda {2r \choose 2} + 8 \lambda p + 4 \lambda r / 5,
\end{split}
\end{equation}
if and only if there is a matching with $p$ sets.

Assume that Eq.~\ref{eq:cond} holds.
Let $Q_j$ be the optimal subgraphs in $F_j$. Write
$S_i$ to be the optimal subgraph in $G_i$ and $S'_i$ to be the optimal subgraph in $G_i'$.
Due to symmetry, we can safely assume that the subgraphs $Q_j$ are all equal.
Next, we claim that the densities and the Jaccard indices of $S_i$ and $S_i'$
cannot reach $q\frac{3}{4}$ and therefore, $Q_j = \set{z_1, z_2, z_3, z_4}$.
The second highest density that any $F_j$ can reach is $q\frac{2}{3}$.
To prove the claim, we argue that the densities and the Jaccard indices of $S_i$ and $S_i'$ can not reach $\Delta = q (\frac{3}{4} - \frac{2}{3}) = \frac{q}{12}$. Since $n \ge 4$, $\Delta > 42r$. There are $2r$ number of density terms for $S_i$ and $S_i'$ with each density less than 1. Since Jaccard index is always less than or equal to 1, their Jaccard terms can not exceed $\lambda (2rq + {2r \choose 2})$. Combining both density and Jaccard terms results in an amount less than $4r$ proving the claim.
Therefore densities and Jaccard indices of $Q_j$ now correspond to $T_1$.

Next, we claim that $S_i$ (and $S_i'$) consists of 4 nodes and 3 edges, two of them being $(z_1, z_2)$ and $(z_1, z_3)$.
Fix $S_i$ with $y = \abs{E(S_i)}$ edges and $x = \abs{S_i}$ nodes. Let $z = \abs{S_i \cap \set{z_1, z_2, z_3, z_4}}$ be the intersection
with  $Q_j$. We can show that the score is 
\[
	\score{\calS} = \frac{y}{x} + 0.55 \frac{z}{x + 4 - z} + R + C,
\]
where $R$ contains the Jaccard terms using $S_i$ and not $Q_j$, and $C$ contains the remaining densities and Jaccard terms not depending on $S_i$.
The first two terms form a fraction with a denominator of at most $n^2$. Consequently, any changes to $x$, $y$, and $z$ change
the first two terms by at least $n^{-4}$.
Note that $R$ contains only $2r - 1$ terms, and due to $\lambda$, we have $R < n^{-4}$. In other words, $\calS$ must optimize the first two terms.
Since there are only $6$ non-singleton nodes in $G_i$, $x \leq 6$. Moreover, $z \leq \min(3, x)$ and $y \leq \min(5, x - 1)$. Enumerating all the possible combinations
show that $x = 4$ and $z = y = 3$ yield optimal score. This is only possible if $S_i$ consists of 4 nodes and 3 edges, two of them being $(z_1, z_2)$ and $(z_1, z_3)$.
Now, densities of $S_i$ (and $S_i'$) and Jaccard indices between $S_i$ (and $S'_i$) and $Q_j$ correspond to $T_2$.

Finally, let us look now at the Jaccard terms between $S_i$ and/or $S'_j$. These terms will constitute $T_3$.
Let $a$ be the number of nodes $v_i$ or $v_i'$ that are included in 3 subgraphs; any such node will yield 3 Jaccard terms of value 1.
Let $b$ be the number of nodes $v_i$, $v_i'$, or $w_i$ that are included in 2 subgraphs; any such node will yield one Jaccard term of value 1.
The remaining Jaccard terms between $S_i$ and $S_i'$ are all of value $3/5$.
In summary, the terms are equal to
\[
	\lambda {2r \choose 2} \frac{3}{5} + \lambda a 3\pr{1 - \frac{3}{5}} + \lambda b \pr{1 - \frac{3}{5}} = \frac{\lambda 3}{5} {2r \choose 2} + \lambda a \frac{4}{5}  + \frac{\lambda (a + b) 2 } 5\quad.
\]
Assume that $a < 2p$. Then since $a + b \leq 2r$ these terms are less than $T_3$, which is a contradiction. Therefore, $a \geq 2p$.
Now, there are $p$ vertices in $\set{v_i}$ or $p$ vertices in $\set{v_i'}$ that are included in 3 subgraphs.
These sets correspond to matchings, and at least one of them will have
$p$ sets.

To prove the other direction, assume there is a matching $\fm{M}$ with $p$ sets.
To form the subgraph sequence,
we first select $(z_1, z_2)$ and $(z_1, z_3)$ in every set. For $u_i \in C_j \in \fm{M}$, we also select $(v_j, z_1)$ in $G_i$ and $(v_j', z_1)$ in $G_i'$.
For an item $u_i$ not covered by $\fm{M}$, we select $(w_i, z_1)$ in $G_i$ and $(w_i, z_1)$ in $G_i'$.
A straightforward calculation shows that this sequence yields the score given in Eq.~\ref{eq:cond}.
\qed
\end{proof}

\section{Algorithms}\label{sec:algorithm}

Since our problem is \np-hard, we need to resort to heuristics.
In this section we present two algorithms that we will use to find good subgraphs.

The first algorithm is as follows. We start with an initial cadidate $\calS$.
This set is either the solution of the densest common subgraph 
or the densest subgraphs of each individual snapshots; we test both and select the best end result.


In order to improve the initial set we employ the strategy used in~\citep{asahiro2000greedily,charikar2000greedy} when approximating the densest subgraph:
Here, the algorithm starts with the whole graph and removes a node with the minimum degree, or equivalently, removes
a node such that the remaining subgraph has the highest density. This is continued until no nodes remained, and among the tested subgraphs
the one with the highest density is selected.

We employ a similar strategy.
For a snapshot $G_i$,
we start with $S_i = V$, and then iteratively remove the vertices so that the score is maximal.
After removing all vertices, we pick the subgraph for $S_i$ which maximizes our objective $\score{\calS; \lambda}$.
We iterate over all snapshots,
we keep on modifying the sets until the algorithm converges. The pseudo-code for this approach is given in Algorithm~\ref{alg:soft-1}.

\begin{algorithm}[t!]
\caption{$\algsoftone(\calG,\lambda, \calS)$, finds subgraphs with good $\score{\cdot; \lambda}$}
\label{alg:soft-1}
    
	
    \While {changes}{
	    \ForEach {$i = 1,\dots, k$} {
	        $C \define V$\;
	        \ForEach{$j = 2, \ldots, \abs{V}$} {
	            $u \define \displaystyle\arg\max_{v \in C} \score{S_1, \ldots, S_{i- 1}, C \setminus \{v\}, S_{i + 1}, \ldots, S_k}$\;
	            $C \define C \setminus \{u\}$\;
    	    }
			$S_i \define $ best tested $C$, if the score improves\; 
        }
    }
	\Return $\calS$\;
\end{algorithm}

Our second algorithm is similar to \algsoftone.
In Algorithm~\ref{alg:soft-1} we consider each snapshot separately and peel off vertices.
In our second algorithm, we initialize each $S_i$ with $V$.
In each iteration, we find a snapshot $S_i$ and a vertex $v$ so that the remaining subgraph sequence is maximized.
We remove the node and continue until no nodes are left.
In the process, we choose the one which maximizes our objective function.
The pseudo-code for this method is given in Algorithm~\ref{alg:soft-2}.

\begin{algorithm}[t!]
\caption{$\algsofttwo(\calG,\lambda)$, finds subgraphs with good $\score{\cdot; \lambda}$}
\label{alg:soft-2}
    $\calS \define S_1, \ldots, S_k$, where $S_i = V$\;
    \While {there are nodes} {
		$u, j \define \displaystyle\arg\max_{v, i \mid v \in S_i} \score{S_1, \ldots, S_{i- 1}, S_i \setminus \{v\}, S_{i + 1}, \ldots, S_k}$\;
		$S_j \define S_j \setminus \set{u}$\;
    }
	\Return best tested $\calS$\;
\end{algorithm}

The bottleneck in both algorithms is finding the next vertex to delete.
Let us now consider how we can speed up this selection.
To this end,
select $S_i$ and let $v \in S_i$. Let us write $\calS'$ to be $\calS$ with $S_i$ replaced with $S_i \setminus \set{v}$.
We can write the score difference between $\calS$ and $\calS'$ as
\begin{equation}
\label{eq:gain}
	\score{\calS'} - \score{\calS} = \frac{\abs{E(S_i)} - \deg v}{\abs{S_i} - 1} - \frac{\abs{E(S_i)}}{\abs{S_i}} + \sum_{j \neq i} \jaccard{S_i \setminus \set{v}, S_j} - \jaccard{S_i, S_j}\quad.
\end{equation}

Let us first consider \algsofttwo.
To find the optimal $v$ and $i$, we will group the nodes in $S_i$ such that the
sum in Eq.~\ref{eq:gain} is equal for the nodes in the same group.
In order to do that we group the nodes based on the following condition: if two nodes $u,v  \in S_i$ belong to the exactly same $S_j$ for each $j$, that is,
$u \in S_j$ if and only if $v \in S_j$, then $u$ and $v$ belong to the same group.
Let us write $\fm{P}$ to be the collection of all these groups (across all $i$).

Select $P \in \fm{P}$. Since the sum in Eq.~\ref{eq:gain} is constant for all nodes in $P$, the node in $P$
maximizing Eq.~\ref{eq:gain} must have the smallest degree.
Thus, we maintain the nodes in $P$ in a priority queue keyed by the degree. We also maintain the difference of the Jaccard indices, the sum in Eq.~\ref{eq:gain}.
In order to maintain the difference, we maintain the sizes of intersection $\abs{S_i \cap S_j}$ and the union $\abs{S_i \cup S_j}$ for all $i$ and $j$.
To find the optimal $v$ and $i$, we find the vertex with the smallest degree in each group, and then compare these candidates among different groups.

This data structure leads to the following running time.

\begin{proposition}
Assume a graph sequence $G_1, \ldots, G_k$ with $n$ nodes and total $m = \sum_i m_i$ edges.
Let $\fm{P}_{ir}$ be the groups of $S_i$ (based on the node memberships in other snapshots) when deleting $r$th node.
Define $\Delta = \max \abs{\fm{P}_{ir}}$.
Then the running time of \algsofttwo is in
\[
	\bigO{nk^2 \Delta + m \log n + k^2 n(k + \log n)} \subseteq \bigO{n^2k^2 + m \log n + k^3 n}\quad.
\]
\label{prop:soft1}
\end{proposition}

\begin{proof}
Finding the best node $u$ and the snapshot $S_i$ requires $\bigO{\sum_i \abs{\fm{P}_{ir}}} \in \bigO{k \Delta }$ time.
Consider deleting $u$ from $S_i$.

Deleting $u$ from its queue requires $\bigO{\log n}$ time.
Upon deletion, we update the degrees of the neighboring nodes in the corresponding queues, in \emph{total} time of $\bigO{m \log n}$.
Updating the intersection and the union sizes requires $\bigO{k}$ time.

We also need to update the gain coming from Jaccard indices for each group $P$. Only one term changes if $P$ is not a subset of $S_i$; there are at most $k\Delta$ such groups.
Otherwise, if $P \subseteq S_i$, then $k-1$ terms change; there are at most $\Delta$ such groups. In summary, we need $\bigO{k \Delta}$ time.

Node $u$ is included in $\bigO{k}$ queues. As we remove $u$ from $S_i$, these queues need to be updated by moving $u$ to the correct queue.
A single such update requires deleting $u$ from its current queue, finding (and possibly creating) the new queue,
and adding $u$ to it. This can be done in $\bigO{k + \log n}$ time.

Combining these times proves the claim.
\qed
\end{proof}

We should point out that the running time depends on $\Delta$, the number of queues in a single snapshot.
This number may be as high the number of vertices, $n$, but ideally $\Delta \ll n$.

The same data structure can be also used \algsoftone.
The only difference is that we do not select optimal $i$; instead $i$ is fixed when looking for the next vertex to delete.
Trivial adjustments to the proof of Prop.~\ref{prop:soft1} imply the following claim.

\begin{proposition}
Assume a graph sequence $G_1, \ldots, G_k$ with $n$ nodes and total $m = \sum_i m_i$ edges.
Let $\fm{P}_{ir}$ be the groups of $S_i$ (based on the node memberships in other snapshots) when deleting $r$th node.
Define $\Delta = \max \abs{\fm{P}_{ir}}$.
Then the running time of a single iteration of \algsoftone is in
\[
	\bigO{nk^2 \Delta + m \log n + k^2 n(k + \log n)} \subseteq \bigO{n^2k^2 + m \log n + k^3 n}\quad.
\]
\label{prop:soft2}
\end{proposition}


\section{Related work}\label{sec:related}
In this section we discuss
previous studies on discovering the densest
subgraph in a single graph, the densest common subgraph over multiple graphs, overlapping densest subgraphs, and other types of density measures. 

\textbf{The densest subgraph:}
Given an undirected graph, finding the subgraph which maximizes density has
been first studied by  Goldberg~\cite{goldberg1984finding} where an exact,
polynomial time algorithm which solves a sequence of min-cut instances is
presented. \citet{asahiro2000greedily} provided a linear time, greedy algorithm
proved to be an  $1/2$-approximation algorithm by
Charikar~\cite{charikar2000greedy}. The idea of the algorithm is that
at each iteration, a vertex with minimum degree is removed, and then the densest
subgraph among all the produced subgraphs is chosen.

Several variants of the densest subgraph problem constrained
on the size of the subgraph~$\abs{S}$ have been studied: finding the densest
$k$-subgraph~($ \abs{S} = k$)~\cite{feige2001dense, khot2006ruling,
asahiro2000greedily},  \emph{at most} $k$-subgraph~($\abs{S} \leq
k$)~\cite{andersen2009finding, khuller2009finding}, and \emph{at least}
$k$-subgraph~($ \abs{S} \geq
k$)~\cite{andersen2009finding,khuller2009finding}. Unlike the densest subgraph
problem, when the size constraint is applied, the densest $k$-subgraph problem
becomes \np-hard~\cite{feige2001dense}. Furthermore,  there is no polynomial
time approximation scheme~(PTAS)~\cite{khot2006ruling}. Approximating the
problem of finding at most $k$-subgraph is  shown as hard as the densest
$k$-subgraph problem  by ~\citet{khuller2009finding}. To find exactly $k$-size
densest subgraph, \citet{bhaskara2010detecting} gave an $\bigO{n^{1/4 +
\epsilon}}$-approximation algorithm for every $\epsilon > 0$ that runs in
$n^{\bigO{1/\epsilon}}$ time. 
\citet{andersen2009finding} provided a linear time $1/3$-approximation algorithm for at least $k$ densest subgraph problem.  

\textbf{The densest common subgraph over multiple graphs:}
\citet{jethava2015finding} extended the densest subgraph problem~(DCS) for the case of multiple graph snapshots.
As a measure the authors' goal was to maximize the minimum density.
Moreover, \citet{semertzidis2019finding} introduced several variants of this
problem by varying the aggregate function of the optimization problem, 
one variant, BFF-AA, is same as the DCS problem discussed in Section~\ref{sec:prel}.
DCS can be solved exactly through a reduction to the
densest subgraph problem, and is consequently polynomial.  The hardness of DCS variants has been
addressed~\cite{charikar_on_finsing_dcs}.

\textbf{Overlapping densest subgraphs of a single graph:}
Finding multiple dense subgraphs in a single graph which allows graphs to be overlapped is studied by adding a hard
constraint to control the overlap of
subgraphs~\cite{balalau2015finding}.
Later, \citet{galbrun2016top} formulated the same problem  adding a penalty in
the objective function for the overlap. 
The difference of our problem to the works of
\citet{balalau2015finding} and \citet{galbrun2016top} is that our goal is to
find a collection of dense subgraphs 
over multiple graph snapshots~(one dense subgraph for each graph snapshot) while they discover a set of dense subgraphs within a single graph. 
Due to this difference, we want to reward similar subgraphs while the authors want to penalize similar subgraphs.


\textbf{Other density measures:}
We use the ratio of edges over the nodes as our measures as it allows us to
compute it efficiently. Alternative measures have been also considered.  One
option is to use the proportion of edges instead, that is, $\abs{E} / {\abs {V}
\choose 2}$. The issue with this measure is that a single edge yields the
highest density of 1. Moreover, finding the largest graph with the edge
proportion of 1 is equal to finding a clique, a classic problem that does not
allow any good approximation~\citep{hastad1996clique}.
As an alternative approach, \citet{tsourakakis2013denser}
proposed finding subgraphs with large score $\abs{E} - \alpha {\abs {V} \choose 2}$.
Optimizing this measure is an \np-hard problem but an algorithm similar to the one given by~\citet{asahiro2000greedily}
leads to an additive approximation guarantee. In similar vein, \citet{tatti2019density} considered
subgraphs maximizing  $\abs{E} - \alpha \abs {V}$ and showed that they form nested structure
similar to $k$-core decomposition.
An alternative measure called
triangle-density has been proposed by~\citet{tsourakakis2015k} as a ratio of triangles 
and the nodes, possibly producing smaller graphs. Like the density, optimizing
this measure can be done in polynomial time.
We leave adopting these measures as a future work.

\section{Experimental evaluation}\label{sec:exp}

\begin{table}[t!]
\setlength{\tabcolsep}{0pt}
\caption{Characteristics of  synthetic datasets. Here, $\abs{V^d}$ and $\abs{V^s}$ give initial number of dense and sparse vertices respectively, $\mean{}{\abs{E}}$ is the expected number of edges, $k$  is the number of snapshots,
$p_d$, $p_s$, and $p_c$ gives the dense, sparse, and cross edge probabilities,
 $J_{min} = \min_{i < j} J(V^d_i, V^d_j)$ is the minimum Jaccard index between ground truth dense sets of vertices, $d_{true}$ is the ground truth density of dense components, $d_{dcs}$ gives the density of densest common subgraph, and $d_{ind}$ gives the sum of densities of locally densest subgraph from each graph snapshot.
}

\label{tab:stats1}
\begin{filecontents*}{c1.csv}
name,n_d,n_s,m,tau,den_pro,spars_pro,cross_pro,true_den,dcs_den,dense_den,min_jac_gnd,file
\dtname{Syn-1},100,900,672.8,10,0.05,0.0005,0.0005,34.67821244,24.26,35.73889802,0.4678899083,file-syn-exp-2
\dtname{Syn-2},100,5000,3850.2,5,0.05,0.0001,0.0001,40.47179788,13.02631579,40.56791081,0.1835334477,file-syn-exp-3
\dtname{Syn-3},120,1200,3922,5,0.06,0.005,0.002,27.40471445,17.61538462,27.61690797,0.4528985507,syn-6-sensitivity
\dtname{Syn-4},250,5000,4709,8,0.03,0.001,0.001,55.78697597,29.64940239,55.97708877,0.2689320388,s-n-2
\dtname{Syn-5},500,3500,12136.28571,7,0.05,0.0003,0.0003,112.1142415,87.26746507,112.1244418,0.5267947422,s-n-1
\dtname{Syn-6},350,3500,32015,5,0.06,0.005,0.002,67.96373302,52.17039106,67.96461663,0.4905913978,syn-8-sensitivity
\end{filecontents*}
\pgfplotstabletypeset[
    begin table={\begin{tabular*}{\textwidth}},
    end table={\end{tabular*}},
    col sep=comma,
	columns = {name,n_d,n_s,m,tau,den_pro,spars_pro,cross_pro,true_den,dcs_den,dense_den,min_jac_gnd},
    columns/name/.style={string type, column type={@{\extracolsep{\fill}}l}, column name=\emph{Dataset}},
    columns/n_d/.style={fixed, set thousands separator={\,}, column type=r, column name=$\abs{V^d}$},
    columns/n_s/.style={fixed, set thousands separator={\,}, column type=r, column name=$\abs{V^s}$},
    columns/m/.style={fixed, set thousands separator={\,}, column type=r, column name=$\mean{}{\abs{E}}$},
    columns/tau/.style={fixed, set thousands separator={\,}, column type=r, column name=$k$},
    columns/lam/.style={fixed, set thousands separator={\,}, column type=r, column name=$\lambda$},
    columns/true_den/.style={fixed, set thousands separator={\,}, column type=r, column name=$d_{true}$}, 
    columns/dcs_den/.style={fixed, set thousands separator={\,}, column type=r, column name=$d_{dcs}$}, 
    columns/dense_den/.style={fixed, set thousands separator={\,}, column type=r, column name=$d_{ind}$}, 
    columns/min_jac_gnd/.style={fixed, set thousands separator={\,}, column type=r, column name=$J_{min}$}, 
    columns/itr/.style={fixed, set thousands separator={\,}, column type=r, column name=$itr$},
    columns/den_pro/.style={fixed,set thousands separator={\,},precision=4, column type=r, column name=$p_{d}$},
    columns/spars_pro/.style={fixed,set thousands separator={\,},precision=4, column type=r, column name=$p_{s}$},
    columns/cross_pro/.style={fixed,set thousands separator={\,},precision=4, column type=r, column name=$p_{c}$},
    every head row/.style={
		before row={\toprule},
			after row=\midrule},
    every last row/.style={after row=\bottomrule},
]
{c1.csv}
\end{table}

\begin{table}[t!]
\setlength{\tabcolsep}{0pt}
\caption{
Computational statistics from the experiments for  synthetic datasets using
\algsoftone and \algsofttwo  algorithms. Here, $\lambda$ is the parameter in
$\score{\cdot; \lambda}$, $i$ is the number of iterations using \algsoftone,
columns
$d_{dis}$ and $\score{}$ are the sum of densities and scores of the discovered sets, $J_{min}$ provide the minimum Jaccard index between discovered sets, columns $\rho$ give the  average Jaccard index between discovered and ground truth sets, and  columns $time$ give the computational time in seconds.
}

\label{tab:stats-syn-soft}
\begin{filecontents*}{syn-soft.csv}
dataset,l_alpha,dis_d_1,dis_d_2,dis_obj_1,dis_obj_2,z2,min_jac_dis_1,min_jac_dis_2,time_1,time_2,avg_jac_1,avg_jac_2,itr_1,new_init
\dtname{Syn-1},0.30,34.841,34.843,43.099,43.090,43.099,0.508,0.497,10,30,0.929,0.930,4,
,0.50,33.542,33.439,48.971,48.970,48.971,0.537,0.537,11,33,0.887,0.883,5,
,0.70,26.431,26.220,56.083,56.076,56.006,0.840,0.862,4,32,0.736,0.733,2,
\dtname{Syn-2},0.3,40.56623994,40.56791108,41.18414136,40.77309645,41.18414136,0.1840277778,0.1840277778,50,92,0.9867704661,0.9860947904,3,
,0.40,40.566,40.566,41.390,41.390,41.390,0.184,0.184,53,84,0.987,0.987,3,
,0.50,40.566,40.566,41.596,41.596,41.596,0.184,0.184,50,87,0.986,0.986,3,
\dtname{Syn-3},0.40,27.58505869,27.60080732,29.59312903,29.58553274,29.59312903,0.4391143911,0.4358974359,16,29,0.9729911563,0.9743407812,3,
,0.50,27.58505869,27.60080732,30.09514661,30.08171409,30.09514661,0.4391143911,0.4358974359,18,26,0.9729911563,0.9743407812,3,
,0.80,27.44214001,27.60226364,31.64622651,31.55850885,31.64622651,0.4552238806,0.4343065693,17,28,0.9759770919,0.976489434,3,
\dtname{Syn-4},0.40,55.74781026,55.8919221,60.37848492,60.4176242,55.88247945,0.2715686275,0.2715686275,66,307,0.9758097821,0.9799004593,2,
,0.50,55.82483829,55.82986047,61.55742797,61.554486,61.55742797,0.2729064039,0.2729064039,132,308,0.9816717382,0.9829339014,4,
,0.60,55.76433157,55.78595627,62.70991151,62.70627001,62.70991151,0.2729064039,0.2729064039,132,285,0.9796133485,0.9782129422,4,
\dtname{Syn-5},0.01,112.1244424,112.1244424,112.2635995,112.2635995,112.2635995,0.5257836198,0.5257836198,73,249,0.9990625258,0.9990625258,2,
,0.40,112.1225599,112.1225599,117.6950848,116.3019536,117.6950848,0.5267947422,0.5267947422,69,251,0.999510678,0.999510678,2,
,0.50,112.1225599,112.1225599,119.0882161,119.0882161,119.0882161,0.5267947422,0.5267947422,103,257,0.999510678,0.999510678,3,
\dtname{Syn-6},0.10,67.96373302,67.96461681,68.61435563,68.61364617,68.61438419,0.4892473118,0.4892473118,42,226,0.9990805883,0.9980713924,2,
,0.80,67.95630902,67.96373302,73.18005294,73.1687139,73.18005294,0.4919137466,0.4892473118,67,252,0.9992110454,0.9990805883,3,
,1.00,67.9523722,67.96220656,74.48631243,74.47656508,74.48631243,0.4925775978,0.4892473118,69,238,0.998816568,0.9996055227,3,
\end{filecontents*}
\pgfplotstabletypeset[
    begin table={\begin{tabular*}{\textwidth}},
    end table={\end{tabular*}},
    col sep=comma,
	columns = {dataset,l_alpha,
	dis_d_1,
	dis_obj_1,
	min_jac_dis_1,
	avg_jac_1,
	time_1,
	itr_1,
	dis_d_2,
	dis_obj_2,
	min_jac_dis_2,
	avg_jac_2,
	time_2},
    columns/dataset/.style={string type, column type={@{\extracolsep{\fill}}l}, column name=\emph{Data}},
    columns/n/.style={fixed, set thousands separator={\,}, column type=r, column name=$\abs{V}$},
    columns/m/.style={fixed, set thousands separator={\,}, column type=r, column name=$Exp[\abs{E}]$},
    columns/time_0/.style={fixed, set thousands separator={\,}, column type=r, column name=$time$},
    columns/time_1/.style={fixed, set thousands separator={\,}, column type=r, column name=$time$},
    columns/time_2/.style={fixed, set thousands separator={\,}, column type=r, column name=$time$},
    columns/l_alpha/.style={fixed, set thousands separator={\,}, column type=r, column name=$\lambda$},
    columns/true_den/.style={fixed, set thousands separator={\,}, column type=r, column name=$d_{true}$}, 
    columns/dcs_den/.style={fixed, set thousands separator={\,}, column type=r, column name=$d_{dcs}$}, 
    columns/dense_den/.style={fixed, set thousands separator={\,}, column type=r, column name=$d_{den}$}, 
    columns/min_jac_dis_0/.style={fixed, set thousands separator={\,}, column type=r, column name=$J_{min}$}, 
    columns/min_jac_dis_1/.style={fixed, set thousands separator={\,}, column type=r, column name=$J_{min}$}, 
    columns/min_jac_dis_2/.style={fixed, set thousands separator={\,}, column type=r, column name=$J_{min}$}, 
    columns/avg_jac_0/.style={fixed, set thousands separator={\,}, column type=r, column name=$\rho$}, 
    columns/avg_jac_1/.style={fixed, set thousands separator={\,}, column type=r, column name=$\rho$}, 
    columns/avg_jac_2/.style={fixed, set thousands separator={\,}, column type=r, column name=$\rho$}, 
    columns/itr_0/.style={fixed, set thousands separator={\,}, column type=r, column name=$itr$},
    columns/itr_1/.style={fixed, set thousands separator={\,}, column type=r, column name=$i$},
    columns/dis_d_0/.style={fixed,set thousands separator={\,}, column type=r, column name=$d_{dis}$},
    columns/dis_d_1/.style={fixed,set thousands separator={\,}, column type=r, column name=$d_{dis}$},
    columns/dis_d_2/.style={fixed,set thousands separator={\,}, column type=r, column name=$d_{dis}$},
    columns/dis_obj_1/.style={fixed,set thousands separator={\,}, column type=r, column name=$\score{}$},
    columns/dis_obj_2/.style={fixed,set thousands separator={\,}, column type=r, column name=$\score{}$},
    every head row/.style={
		before row={\toprule
		& &
		\multicolumn{6}{l}{\algsoftone} &
		\multicolumn{5}{l}{\algsofttwo} \\
		\cmidrule(r){3-8}
		\cmidrule(r){9-13}
		},
			after row=\midrule},
    every last row/.style={after row=\bottomrule},
]
{syn-soft.csv}
\end{table}

The goal in this section is to experimentally evaluate  our
algorithms. We first generate several synthetic datasets  and plant  dense subgraph components, in each snapshot and test how well our algorithms discover the ground truth.
Next we study the performance of the algorithm on real-world temporal datasets in terms of running time.
We compare our results with the solutions obtained with the densest common subgraph~\cite{semertzidis2019finding} and the sum of densities of locally dense subgraphs~\cite{goldberg1984finding}. Finally, we present results of a case study.

We implemented the algorithms in Python\footnote{The source code is available at \url{https://version.helsinki.fi/dacs/}.
\label{foot:code}} and performed the experiments using a 2.4GHz Intel Core i5 processor and 16GB RAM.

\textbf{Synthetic datasets:}
Next, we explain in detail how the synthetic datasets are generated, the statistics and the related parameters are given in Table~\ref{tab:stats1}. 

Each dataset consists of $k$ graphs given as $\{G_1,\ldots,G_k\}$.
We split the vertex set into dense and sparse components $V^d$ and $V^s$.
To generate the $i$th snapshot we create two components  $V^d_i$ and $V^s_i$
by starting from $V^d$ and $V^s$ and moving nodes from $V^s$ to $V^d$ with a probability of $\eta_i$.
The probability $\eta_i$ is selected randomly for each snapshot from a uniform distribution $[0.01, 0.09]$.
Once the vertices are generated, we sample the edges using a stochastic block model, with the edge probabilities
being $p_d$, $p_s$, and $p_c$ for dense component, sparse component, and cross edges, respectively.
We created $6$ such synthetic datasets in total to test our algorithms.

\textbf{Results of synthetic datasets:}
We report our results in Table~\ref{tab:stats-syn-soft}.

First, we  observe that the discovered density values $d_{dis}$ approximately match each other, that is, both \algsoftone and \algsofttwo perform equally well in terms of the densities. Similar
result holds also for the scores $\score{\cdot; \lambda}$ and 
minimum Jaccard coefficients $J_{min}$.
However, \algsoftone runs faster than  \algsofttwo. This is probably due to the fact that \algsofttwo takes more time to select the next vertex to delete, which is the bottleneck in both algorithms despite of having same asymptotic time complexity per iteration in \algsoftone and overall time complexity in \algsofttwo.
Let us now compare the discovered sets  to the ground truth, given in columns $\rho$. We can see both algorithms gives similar values which indicates equally good performance of \algsoftone and \algsofttwo.

Our next step is to study the effect of the input parameter $\lambda$. First, we observe Figure~\ref{fig:lam-jac} which demonstrates $\rho$ as a function of $\lambda$.
In Figure~\ref{fig:lb}, we see that $\rho$ gradually decreases as we increase $\lambda$. This is due to the fact that when we increase the weight of constraint part of $\score{}$, the algorithms try to  find dense sets with higher Jaccard coefficients which eventually forces to deviate from its ground truth. Furthermore, if we set $\lambda = 2$, we can see a drastic change in $\rho$. 

Let us now consider Figure~\ref{fig:db} which demonstrates how the discovered sum of densities  change with respect to  $\lambda$.
We see the decreasing trend showing that second term in the objective function starts to dominate with the increase of $\lambda$.

Next, we study how the score function $\score{}$ behaves over $\lambda$ shown in Figure \ref{fig:obja}. We can observe that both \algsoftone and \algsofttwo have an increasing trend when we increase $\lambda$ from $0.3$ to $2.3$. This is expected as larger $\lambda$ should yield larger scores.

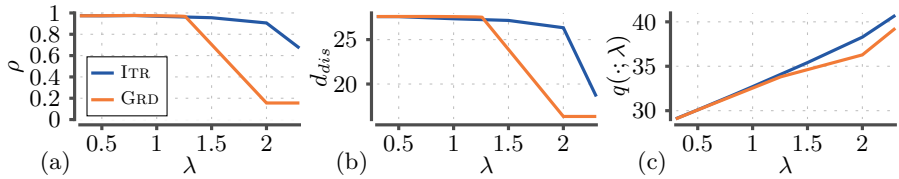
\begin{figure}[t!]
\begin{subcaptiongroup}
\phantomcaption\label{fig:lb}
\phantomcaption\label{fig:db}
\phantomcaption\label{fig:obja}
\begin{center}
\begin{tabular}{llll}

\begin{tikzpicture}
\begin{axis}[xlabel={$\lambda$}, ylabel= {$\rho$},
    width = 4.5cm,
    height = 3cm,
    xmin = 0.3,
    xmax = 2.3,
    ymin = 0,
    ymax = 1,
    scaled y ticks = false,
    cycle list name=yaf,
    yticklabel style={/pgf/number format/fixed},
    no markers,
	legend entries = {\algsoftone, \algsofttwo},
	legend pos = south west
]
\addplot table [x=x, y=y, col sep=comma] {fig-1-b.csv};
\addplot table [x=x, y=y, col sep=comma] {fig-1-c.csv};
\pgfplotsextra{\yafdrawaxis{0.3}{2.3}{0}{1}}
\end{axis}
\node[anchor=north east] at (0, -0.3) {(a)};
\end{tikzpicture}&

\begin{tikzpicture}
\begin{axis}[xlabel={$\lambda$}, ylabel= {$d_{dis}$},
    width = 4.5cm,
    height = 3cm,
    xmin = 0.3,
    xmax = 2.3,
    ymin = 16,
    ymax = 28,
    scaled y ticks = false,
    cycle list name=yaf,
    yticklabel style={/pgf/number format/fixed},
    no markers,
]
\addplot table [x=x, y=y1, col sep=comma] {fig-1-b.csv};
\addplot table [x=x, y=y1, col sep=comma] {fig-1-c.csv};
\pgfplotsextra{\yafdrawaxis{0.3}{2.3}{16}{28}}
\end{axis}
\node[anchor=north east] at (0, -0.3) {(b)};
\end{tikzpicture}&
\begin{tikzpicture}
\begin{axis}[xlabel={$\lambda$}, ylabel= {$\score{\cdot; \lambda}$},
    width = 4.5cm,
    height = 3cm,
    xmin = 0.3,
    xmax = 2.3,
    ymin = 29,
    ymax = 41,
    scaled y ticks = false,
    cycle list name=yaf,
    yticklabel style={/pgf/number format/fixed},
    no markers,
]
\addplot table [x=x, y=y2, col sep=comma] {fig-1-b.csv};
\addplot table [x=x, y=y2, col sep=comma] {fig-1-c.csv};
\pgfplotsextra{\yafdrawaxis{0.3}{2.3}{29}{41}}
\end{axis}
\node[anchor=north east] at (0, -0.3) {(c)};
\end{tikzpicture}

\end{tabular}
\end{center}
\end{subcaptiongroup}
\caption{
Average Jaccard index to the ground truth $\rho$ as a function of $\lambda$ in Figure \ref{fig:lb}.
Discovered density $d_{dis}$ as a function of $\lambda$ 
in Figure~\ref{fig:db}.
Scores $\score{\cdot; \lambda}$ as a function of $\lambda$
in Figure~\ref{fig:obja}. This experiment was performed using \dtname{Syn-3} dataset.
}
\label{fig:lam-jac}
\end{figure}

\begin{table}[t!]
\setlength{\tabcolsep}{0pt}
\caption{
Characteristics of  real-world datasets. Here, $\abs{V}$  gives the number of
vertices, $\abs{E}$ is the expected number of edges, $k$  is the number of
snapshots, $d_{dcs}$ gives the density of densest common subgraph, and $d_{ind}$
gives the sum of densities of locally densest subgraph from each graph snapshot. 
}

\label{tab:stats4}
\begin{filecontents*}{real-data.csv}
dataset,n,m,tau,dcs_den,dense_den
\dtname{Twitter-\#},806,101.2,15,7.538461538,38.79402685
\dtname{Enron},1079,23.20218579,183,43.47058824,185.3437388
\dtname{Facebook},4117,83.13461538,104,14,88.64638591
\dtname{Students},889,43.68032787,122,24.14705882,117.9763122
\dtname{Twitter-user},4605,109.1935484,93,23,90.62829816
\dtname{Tumblr},1980,65.30337079,89,36.66666667,103.9794686
\end{filecontents*}
\pgfplotstabletypeset[
    begin table={\begin{tabular*}{\textwidth}},
    end table={\end{tabular*}},
    col sep=comma,
	columns = {dataset,n,m,tau,dcs_den,dense_den},
    columns/dataset/.style={string type, column type={@{\extracolsep{\fill}}l}, column name=\emph{Data}},
    columns/n/.style={fixed, set thousands separator={\,}, column type=r, column name=$\abs{V}$},
    columns/m/.style={fixed, set thousands separator={\,}, column type=r, column name=$\abs{E}$},
    columns/tau/.style={fixed, set thousands separator={\,}, column type=r, column name=$k$},
    columns/time_0/.style={fixed, set thousands separator={\,}, column type=r, column name=$t_1$},
    columns/time_1/.style={fixed, set thousands separator={\,}, column type=r, column name=$t_2$},
    columns/time_2/.style={fixed, set thousands separator={\,}, column type=r, column name=$t_3$},
    columns/l_alpha/.style={fixed, set thousands separator={\,}, column type=r, column name=$\lambda$},
    columns/true_den/.style={fixed, set thousands separator={\,}, column type=r, column name=$d_{true}$}, 
    columns/dcs_den/.style={fixed, set thousands separator={\,}, column type=r, column name=$d_{dcs}$}, 
    columns/dense_den/.style={fixed, set thousands separator={\,}, column type=r, column name=$d_{ind}$}, 
    columns/min_jac_dis_0/.style={fixed, set thousands separator={\,}, column type=r, column name=$j_{1}$}, 
    columns/min_jac_dis_1/.style={fixed, set thousands separator={\,}, column type=r, column name=$j_{2}$}, 
    columns/min_jac_dis_2/.style={fixed, set thousands separator={\,}, column type=r, column name=$j_{3}$}, 
    columns/avg_jac_0/.style={fixed, set thousands separator={\,}, column type=r, column name=$j_{1}$}, 
    columns/avg_jac_1/.style={fixed, set thousands separator={\,}, column type=r, column name=$j_{2}$}, 
    columns/avg_jac_2/.style={fixed, set thousands separator={\,}, column type=r, column name=$j_{3}$}, 
    columns/itr_0/.style={fixed, set thousands separator={\,}, column type=r, column name=$i_1$},
    columns/itr_1/.style={fixed, set thousands separator={\,}, column type=r, column name=$i_2$},
    columns/dis_d_0/.style={fixed,set thousands separator={\,}, column type=r, column name=$d_1$},
    columns/dis_d_1/.style={fixed,set thousands separator={\,}, column type=r, column name=$d_2$},
    columns/dis_d_2/.style={fixed,set thousands separator={\,}, column type=r, column name=$d_3$},
    columns/dis_obj_1/.style={fixed,set thousands separator={\,}, column type=r, column name=$o_2$},
    columns/dis_obj_2/.style={fixed,set thousands separator={\,}, column type=r, column name=$o_3$},
    every head row/.style={
		before row={\toprule},
			after row=\midrule},
    every last row/.style={after row=\bottomrule},
]
{real-data.csv}
\end{table}

\begin{table}[t!]
\setlength{\tabcolsep}{0pt}
\caption{
Computational statistics from the experiments for  real-world datasets. 
Here, $\lambda$ is the parameter in
$\score{\cdot; \lambda}$, $i$ gives the number of iterations using \algsoftone,
columns
$d_{dis}$ are the discovered sum of densities,  columns $\score{}$ are the discovered scores, and  columns $time$ give the computational time in seconds.
}

\label{tab:real-soft}
\begin{filecontents*}{real-soft.csv}
dataset,l_alpha,dis_d_1,dis_d_2,new,dis_obj_1,dis_obj_2,min_jac_dis_1,min_jac_dis_2,time_1,time_2,itr_1
\dtname{Twitter-\#},0.3,37.33888889,37.6547619,40.51679468,40.51679468,38.12919423,0,0,2.808625937,5.791913986,5
,0.5,14.14285714,31.01031746,54.53571429,54.53571429,42.36820769,0.1428571429,0,3.061551809,6.382355928,7
,0.7,11,16.8047619,74.7,74.7,56.17561472,0.3333333333,0,2.618840218,6.22947216,5
,0.8,11,9.333333333,83.8,83.8,75.46666667,0.3333333333,0,2.785507202,6.878108025,6
\dtname{Enron},0.05,130.4349697,122.0809524,335.7449507,357.2928071,343.0230014,0,0,72.02136397,816.5367939,5
,0.1,111.875013,62.08333333,552.6889385,593.8688115,589.9583333,0,0,148.9045451,826.2805502,10
,0.5,95.87881892,20,2306.503696,2662.7012,3139.25,0,0,81.56071901,854.4284482,6
,5,88.65937447,20,23695.37381,25254.80595,31212.5,0,0,98.86982393,774.388128,7
\dtname{Facebook},0.1,52.85905325,56.02630039,110.437508,126.9788644,85.63793123,0,0,214.7383478,9461.593565,7
,0.5,36.39561404,25.47074489,311.9773119,446.0874639,344.601335,0,0,148.9367909,10193.95629,6
,0.7,31.25433455,23.41024531,657.9176846,657.9176846,475.8332335,0,0,185.9142609,9781.904757,7
,1,27.41190476,22.48071697,648.4498973,922.4528139,659.3237125,0,0,136.724452,9034.351239,6
\dtname{Students},0,108.901692,102.6952129,108.901692,108.901692,102.6952129,0,0,124.3819127,2787.795456,4
,0.2,46.98716422,33.16699689,526.1863196,526.1863196,567.8440698,0,0,100.1651509,2178.281251,5
,0.5,46.17084027,33.16699689,1258.408709,1258.408709,1369.859679,0,0,79.65233374,2256.262814,4
,0.8,44.50080198,33.16699689,2020.812564,2020.812564,2171.875289,0,0,78.60893488,2125.67154,4
\dtname{Twitter-user},0.01,79.2083574,63.5255474,81.60932521,81.60932521,68.53018669,0,0,572.2469997,7524.473368,8
,0.1,12.03968254,13.89959508,269.0007937,269.0007937,212.2772151,0,0.1176470588,260.3086672,6422.460369,9
,0.2,11.48888889,11.64189506,545.9302381,545.9302381,423.0972076,0,0.1176470588,149.9142139,6770.163138,5
,0.5,11.48888889,10.82634991,1347.592262,1347.592262,1040.046885,0,0.1176470588,150.080363,6847.677817,5
\dtname{Tumblr},0.1,71.32619048,66.47261905,316.0730952,316.0730952,302.8260317,0,0,50.23071408,1229.976281,4
,0.5,64.36666667,62.96904762,1456.466667,1456.466667,1291.815476,0,0,49.30568504,1151.847987,4
,0.7,59.25,62.96904762,2182.291667,2182.291667,1783.354048,0,0,48.96006393,1217.76194,4
,1,59.16666667,62.96904762,3092.083333,3092.083333,2520.661905,0,0,48.37057424,1252.051439,4
\end{filecontents*}
\pgfplotstabletypeset[
    begin table={\begin{tabular*}{\textwidth}},
    end table={\end{tabular*}},
    col sep=comma,
	columns = {dataset,l_alpha,dis_d_1,dis_obj_1,time_1,itr_1,dis_d_2,dis_obj_2,time_2},
    columns/dataset/.style={string type, column type={@{\extracolsep{\fill}}l}, column name=\emph{Data}},
    columns/n/.style={fixed, set thousands separator={\,}, column type=r, column name=$\abs{V}$},
    columns/m/.style={fixed, set thousands separator={\,}, column type=r, column name=$\abs{E}$},
    columns/tau/.style={fixed, set thousands separator={\,}, column type=r, column name=$\tau$},
    columns/time_0/.style={fixed, set thousands separator={\,}, column type=r, column name=$time$},
    columns/time_1/.style={fixed, set thousands separator={\,}, column type=r, column name=$time$},
    columns/time_2/.style={fixed, set thousands separator={\,}, column type=r, column name=$time$},
    columns/l_alpha/.style={fixed, set thousands separator={\,}, column type=r, column name=$\lambda$},
    columns/true_den/.style={fixed, set thousands separator={\,}, column type=r, column name=$d_{true}$}, 
    columns/dcs_den/.style={fixed, set thousands separator={\,}, column type=r, column name=$dcs$}, 
    columns/dense_den/.style={fixed, set thousands separator={\,}, column type=r, column name=$d_{l}$}, 
    columns/min_jac_dis_0/.style={fixed, set thousands separator={\,}, column type=r, column name=$j_{1}$}, 
    columns/min_jac_dis_1/.style={fixed, set thousands separator={\,}, column type=r, column name=$j_{2}$}, 
    columns/min_jac_dis_2/.style={fixed, set thousands separator={\,}, column type=r, column name=$j_{3}$}, 
    columns/avg_jac_0/.style={fixed, set thousands separator={\,}, column type=r, column name=$j_{1}$}, 
    columns/avg_jac_1/.style={fixed, set thousands separator={\,}, column type=r, column name=$j_{2}$}, 
    columns/avg_jac_2/.style={fixed, set thousands separator={\,}, column type=r, column name=$j_{3}$}, 
    columns/itr_0/.style={fixed, set thousands separator={\,}, column type=r, column name=$i_1$},
    columns/itr_1/.style={fixed, set thousands separator={\,}, column type=r, column name=$i$},
    columns/dis_d_0/.style={fixed,set thousands separator={\,}, column type=r, column name=$d_1$},
    columns/dis_d_1/.style={fixed,set thousands separator={\,}, column type=r, column name=$d_{dis}$},
    columns/dis_d_2/.style={fixed,set thousands separator={\,}, column type=r, column name=$d_{dis}$},
    columns/dis_obj_1/.style={fixed,set thousands separator={\,}, column type=r, column name=$\score{}$},
    columns/dis_obj_2/.style={fixed,set thousands separator={\,}, column type=r, column name=$\score{}$},
    every head row/.style={
		before row={\toprule
		& &
		\multicolumn{4}{l}{\algsoftone} &
		\multicolumn{3}{l}{\algsofttwo} \\
		\cmidrule(r){3-6}
		\cmidrule(r){7-9}
		},
			after row=\midrule},
    every last row/.style={after row=\bottomrule},
]
{real-soft.csv}
\end{table}

\textbf{Real-world datasets:} 
We consider $6$ publicly available, real-world datasets. The details of the  datasets  are shown in Table~\ref{tab:stats4}.
\dtname{Twitter-\#}~\cite{tsantarliotis2015topic}\footnote{\url{https://github.com/ksemer/BestFriendsForever-BFF-}\label{foot:gitbff}}  is a hashtag network where nodes correspond to hashtags and edges corresponds to the
interactions  where  two hashtags appear in a tweet. This dataset contains $15$ such daily graph snapshots in total. 
\dtname{Enron}\footnote{\url{http://www.cs.cmu.edu/~./enron/}} is a popular dataset which contains email communication data within senior management of Enron company.  It contains $183$ daily snapshots in which daily email count is at least $5$.
\dtname{Facebook}~\cite{viswanath2009evolution}\footnote{\url{https://networkrepository.com/fb-wosn-friends.php} } is a network of Facebook users in  New Orleans regional
community.  It contains a set of facebook wall posts  among these users during 9th of June to 20th of August, 2006.
\dtname{Students}\footnote{\url{https://toreopsahl.com/datasets/\#online_social_network}} is an online message
network at the University of California, Irvine. It spans over $122$ days.
\dtname{Twitter-user}~\cite{rozenshtein2020finding}\footnote{\url{https://github.com/polinapolina/segmentation-meets-densest-subgraph/tree/master/data}} is a network of twitter users in Helsinki 2013.
It contains  a set of tweets which appear each others' names.
\dtname{Tumblr} ~\cite{leskovec2009meme}\footnote{\url{http://snap.stanford.edu/data/memetracker9.html}} contains 
phrase or quote mentions appeared in blogs and news media.  It contains author and meme interactions  of users over $3$ months from February to April in 2009.

\textbf{Results of real-world datasets:}
We report the results obtained from the experiments  with real-world datasets  in Table \ref{tab:real-soft}.

First, we compare scores $\score{}$  obtained using \algsoftone and \algsofttwo.
As we can see, apart from few cases in \dtname{Enron} and \dtname{Students} datasets, \algsoftone  achieves greater score than \algsofttwo.
Furthermore, we can observe in  $time$ columns, \algsoftone runs faster than \algsofttwo.
Next, let us observe column $i$ which gives number of iterations with \algsoftone algorithm. We can see that we have at most $10$ number of iterations which is reasonable to deal with real-world datasets.

As we compare columns $d_{dis}$ which show discovered densities,
we can occasionally see approximately similar values but also deviations. 
Next, let us observe the effect of $\lambda$ on $d_{dis}$ and $\score{}$.
Based on $\lambda$,  the ratio between $d_{dis}$  and  $\score{}$ tends to
change. For example, in general, when $\lambda$ is lowered, $d_{dis}$  tends to
increase while  $\score{}$  decreases, as expected. 

\begin{table}[t!]
\caption{Twitter hash tags discovered for \dtname{Twitter-8} dataset. 
}
\label{tab:twitter}
\begin{tabular*}{\textwidth}{p{12cm}}
\toprule
DCS: abudhabigp, fp1, abudhabi, guti, f1, pushpush, skyf1, hulk, allowin, bottas, kimi, fp3, fp2, unleashthehulk, density: \pgfmathprintnumber[precision=2]{7.071428571428571}\\

\midrule
\algsoftone algorithm: $\lambda = 0.8$, density: \pgfmathprintnumber[precision=2]{12.227272727272727}, objective: \pgfmathprintnumber[precision=2]{21.758887778887782} \\
\midrule
Day 1:  indiangp, skyf1, kimi, f1 \\
Day 2:   abudhabigp, skyf1, f1 \\
Day 3:  kimi, skyf1, abudhabigp, f1 \\
Day 4:   abudhabigp, skyf1, f1 \\
Day 5:   abudhabigp, english, arabic, spanish, french, danish, swedish, f1, endimpunitybh, skyf1, bahrain \\
Day 6:   abudhabigp, fp1, abudhabi, guti, f1, skyf1, hulk, allowin, bottas, kimi, fp3, fp2 \\
Day 7:   abudhabigp, abudhabi, guti, f1, pushpush, skyf1, hulk, allowin, bottas, kimi, fp3, quali \\
Day 8:  abudhabigp, skyf1, f1 \\
\bottomrule
\end{tabular*}
\end{table}

\textbf{Case study using Twitter-8 dataset:}
In this section, we present a case-study and analyze the result which illustrates trending twitter hash tags over a span of $8$ days, under a Jaccard constrained environment.

\dtname{Twitter-8} contains  a hashtag network from November, 2013.
We prepared this dataset by extracting  first $8$ daily graph snapshots from
\dtname{Twitter-\#} dataset.  Here, each node of the graph represents a
specific hashtag.  As seen in the tags from Table \ref{tab:twitter}, Formula-1
racing car event which occurred on Abu Dhabi has been trending during the
period. 
We set  $\lambda = 0.8$ and find different sets of dense subgraphs for each day using \algsoftone. 
On Day 1, tags \textit{indeangp}, \textit{skyf1}, \textit{kimi}, and \textit{f1}  have been 
added to the dense hashtag collection whereas on Day 2, the tag \textit{kimi} who is a Finnish racing legend and the tag \textit{indeangp} have been removed from trending list.
On Day 3, \textit{kimi}  has been re-entered into the trending list and the tag \textit{indeangp} has been replaced by  \textit{abudhabigp}.
On Day 6, the tags like \textit{bottas} (another Finnish racing car driver), \textit{fp2} (Free practice 2), \textit{fp3} (Free practice 3), and etc have been added which indicates that additional racing car event
related tags are trending.
On Day 5, we can observe more new tags like \textit{bahrain},
\textit{english}, \textit{arabic}, \textit{french}, \textit{danish}, and
\textit{swedish} have been appeared which seems not directly related to racing
car event. Moreover, new dense collection gives a higher density of $12.22$ with respect to the DCS density of $7.07$.

\section{Concluding remarks}\label{sec:conclusions}

We introduced a novel Jaccard weighted, dense subgraph discovery problem (\problemsoft) for
graphs with multiple snapshots. Here, our goal was to find a dense subset of
vertices from each graph snapshot such that the sum of densities 
and the similarity between the snapshots is maximized.

We proved that our problem is \np-hard, and  
designed an iterative algorithm which runs in  $\bigO{n^2k^2 + m \log n + k^3
n}$ time per iteration and a greedy algorithm which runs in
$\bigO{n^2k^2 + m \log n + k^3 n}$ time.

We experimentally showed  that the number of iterations was low in iterative
algorithm, and that the
algorithms could find the ground truth using synthetic datasets and could discover dense collections in real world datasets. 
We also studied the effect of our user parameter $\lambda$. 
Finally, we performed a case study showing interpretable results.

The paper introduces several interesting directions for future work.  In this
paper, we enforced the pairwise Jaccard constraint between all available pairs
of snapshots. However, we can relax this constraint further by letting only a
portion of sets which lies within a specific window to assure the Jaccard
similarity constraint which may lead to future work.  Another possible
direction is adopting different type of density for our problem setting.

\bibliographystyle{splncs04nat}
\bibliography{bibliography}

 \appendix

\end{document}